\documentclass{vldb}

\usepackage{balance} 
\usepackage{color}
\usepackage{listings}
\usepackage{algorithm}
\usepackage{algpseudocode}
\usepackage{graphicx}

\newcommand{\bars}[1]{\lvert #1 \rvert}
\newcommand{\abs}[1]{\lvert #1 \rvert}
\newcommand{\pn}[1]{\left ( #1 \right )}
\newcommand{\floor}[1]{\left \lfloor  #1 \right \rfloor }

\DeclareMathOperator{\out}{OUT}
\DeclareMathOperator{\inn}{IN}
\newcommand{\outd}[1]{\bars{\out(#1)}}
\newcommand{\ind}[1]{\bars{\inn(#1)}}

\newtheorem{thm}{Theorem}

\begin{document}
\title{Personalized PageRank to a Target Node}

\numberofauthors{2}
\author{
\alignauthor
Peter Lofgren\\
       \affaddr{Stanford University}\\
       \email{plofgren@stanford.edu}
\alignauthor
Ashish Goel\\
       \affaddr{Stanford University}\\
       \email{ashishg@stanford.edu}
}

\date{\today}

\maketitle

\begin{abstract}

Personalalized PageRank uses random walks to determine the importance or authority of nodes in a graph from the point of view of a given source node.  Much past work has considered how to compute personalized PageRank from a given source node to other nodes.  In this work we consider the problem of computing personalized PageRanks to a given target node from all source nodes.  This problem can be interpreted as finding who supports the target or who is interested in the target.

We present an efficient algorithm for computing personalized PageRank to a given target up to any given accuracy.  
We give a simple analysis of our algorithm's running time in both the average case and the parameterized worst-case.  We show that for any graph with $n$ nodes and $m$ edges, if the target node is randomly chosen and the teleport probability $\alpha$ is given, the algorithm will compute a result with $\epsilon$ error in time  $O\pn{\frac{1}{\alpha \epsilon} \pn{\frac{m}{n} + \log(n)}}$.  
This is much faster than the previously proposed method of computing personalized PageRank separately from every source node, and it is comparable to the cost of computing personalized PageRank from a single source. We present results from experiments on the Twitter graph which show that the constant factors in our running time analysis are small and our algorithm is efficient in practice.

\end{abstract}

\section{Note on Related Work}
After we posted this work, we became aware of the related work \cite{andersen2007local, Andersen:2008:RPL:1451983.1452000}.  It includes an algorithm similar to the one we (independently) discovered. However, our work makes the following novel contributions.  We analyze the algorithm under a more detailed parameterization which includes the in-degree of nodes.  We use a priority queue to obtain a dependence on $\epsilon$ of $O \pn{m \log \pn{\frac{1}{\epsilon}}}$, showing that the running time tends toward the running time of power iteration as $\epsilon$ tends to 0.  Finally, we present detailed experiments to determine the running time of this algorithm on the Twitter graph.

\section{Introduction}
Personalized PageRank is a random-walk based method of modeling how nodes are related in a graph like a social network, the web graph, or a citation graph.  It has been used in a variety of application including personalized search \cite{jeh2003scaling}, link prediction \cite{liben2007link,bahmani2010fast}, link-spam detection \cite{benczur2005spamrank}, and graph partitioning \cite{andersen2006local}.  Previous work has considered how to compute personalized PageRank from a single source node.  In this work, we consider the problem of computing personalized PageRank to a single target node from all source nodes.  More precisely, given a node $v$ in a directed (or undirected) graph $G=(V,E)$, we would like to approximate the personalized PageRanks $\pi(u,v)$ from all nodes $u \in V$ to the target node $v$.  We define the personalized PageRank $\pi(u,v)$ from a node $u$ to a node $v$ to be the fraction of time we spend at $v$ on a random walk from $u$, where after each step we stop with a given probability $\alpha$.  Note that this is different from reversing the edges and computing personalized PageRank from a single source $v$.  If edges represent interest, this problem can be interpreted as finding the nodes $u$ which are interested in $v$, or if edges represent support this problem can be interpreted as finding nodes which support $v$.

This problem has several applications.  In a social network, for example, whenever $v$ produces content, we might want to find the nodes $u$ with $\pi(u,v)$ above some threshold and add the content to each such $u$'s feed.  Or an advertiser $v$ on a social network might want to give special offers to the nodes $u$ which are most interested in it.  The use of personalized PageRank in recommendation and trust systems is discussed in \cite{andersen2008trust}.  On the web graph, this problem has been considered before.  In \cite{benczur2005spamrank}, the first phase of the authors' algorithm to detect when a web page $v$ is benefiting from link-spam is to compute the set of nodes $u$ with a high value of $\pi(u,v)$. 

The simplest solution to this problem is to compute personalized PageRanks from every source node $u$ using known methods like Monte Carlo \cite{fogaras2005towards, bahmani2010fast} or power iteration \cite{page1999pagerank}.  This is the solution proposed in \cite{benczur2005spamrank}, the only previous work on this problem (prior to \cite{andersen2007local, Andersen:2008:RPL:1451983.1452000}).  In \cite{benczur2005spamrank}, the cost of computing personalized PageRank from every source using Monte Carlo is amortized because there are a large number of target nodes $v$.    However, this simple solution requires $n$ Monte Carlo computations  even for a single target $v$.  As shown in \cite{fogaras2005towards} using the Chernoff bounds, computing an $\epsilon$ approximation to $\pi(u,v)$ with high probability from a single source $u$ to all $v$ takes $O \pn{ \frac{1}{\epsilon^2} \log\pn{n}}$ time.  Thus even for a single target this approach would take $O \pn{n \frac{1}{\epsilon^2} \log\pn{n}}$ time.  The challenge we address is finding an algorithm which can find all nodes $u$ that have high values of $\pi(u,v)$ without doing work linear in $n$.

We present an algorithm which, given $v$, approximates $\pi(u,v)$ for all $u$ to within a given additive error without needing to visit all the nodes.  Our method is to start at the target node $v$ and propagate updated estimates of $\pi(u,v)$ backwards along edges.  In power iteration, every node propagates its current value in every step.  The key idea of our algorithm is to maintain a priority queue and only propagate the value of the node whose value has changed the most since its value was last propagated.  It is very simple to implement: the entire algorithm is shown in Algorithm \ref{fig:algorithm}.  We prove that it is efficient both in an average-case and a parameterized sense.  We also present experiments on part of Twitter's graph to show that it is efficient in practice.  In this work we assume that a single processor is used for each target $v$ and that the graph is stored in local or distributed RAM.  If there are multiple targets, parallelism can be achieved by assigning different targets to different processors.  

The contributions of this work are the following: 
\begin{itemize}
\item In section \ref{sec:algorithm} we present a simple algorithm for computing personalized PageRanks to a target node up to any given additive error.  In section \ref{sec:error_analysis} we analyze the approximation error of the algorithm and prove it is correct.

\item In section \ref{sec:average_runtime}, we show that for an arbitrary graph $G$ and a target node $v$ chosen uniformly at random, our algorithm runs in time 
\[  O\pn{\frac{1}{\alpha \epsilon} \pn{\frac{m}{n} + \log(n)}} \] 
where $\frac{m}{n}$ is the average degree of a node, $\epsilon$ is the desired additive error, and $\alpha$ (typically between 0.1 and 0.2 in practice) is the probability of stopping after each step of the walk.  This is comparable to the cost of running Monte Carlo from a single source node, $O\pn{ \frac{1}{\epsilon^2} \log(n)}$ for high success probability, and it is much less than the cost of running Monte Carlo from every source node, $O\pn{ n \frac{1}{\epsilon^2} \log(n)}$. 
\item In section \ref{sec:parameterized_runtime} we show that for an arbitrary graph $G$ and arbitrary target node $v$, our algorithm runs in time
\[ O \pn{\frac{D_v(\alpha \epsilon)}{\alpha} \log \pn{\frac{1}{\epsilon \alpha}}} \]
 where $D_v(\alpha \epsilon) = \sum_{u:\pi(u,v) > \alpha \epsilon} \pn{\ind{u} + \log(n)}$ is a parameter which captures how difficult the problem is for $v$.  This shows the asymptotic dependence on $\epsilon$ is $O\pn{\log \pn{\frac{1}{\epsilon}}}$, and not $O\pn{\frac{1}{\epsilon^2}}$ as it is for Monte Carlo from a single a single source. Even if $\epsilon$ is small enough that we must consider the entire graph,  $D_v(\alpha \epsilon) \leq m + n \log(n) \leq 2m $ for graphs with $ n \log(n) \leq m$.  For such graphs as $\epsilon$ goes to zero the algorithm degrades gracefully to the asymptotic performance of power-iteration, $O \pn{m \log \pn{\frac{1}{\epsilon}}}$.  Thus for larger $\epsilon$ we get the benefit of only exploring a small set of nodes with high personalized PageRank to the target, while for small $\epsilon$ the running time is still comparable to the cost of running power iteration.

\item In section \ref{sec:experiments} we present results from an experiment on part of the Twitter graph with 5.3 million nodes and 380 million edges.  We find that our error analysis is tight and that $D_v(\alpha \epsilon)$ is an accurate paramterization of the running time.  As one example, we find that for a approximation of $\epsilon=10^{-5}$, the priority queue algorithm takes 1.2 seconds while power iteration takes 410 seconds to achieve additive error $\epsilon$ on the same machine.  This shows that the local nature of the algorithm can give significant savings.  
\end{itemize}

\section{Related Work}
(See note on related work in section 1.)

Personalized PageRank was first suggested in the original PageRank paper \cite{page1999pagerank}, and much follow up work has considered how to compute it efficiently.  Our approach of propagating estimate updates is similar to the approach taken by Jeh and Widom \cite{jeh2003scaling} and Berkin  \cite{berkhin2006bookmark} to compute personalized PageRank from a single source.  
Our equation \eqref{eq_power_iteration} appears as equation (10) in \cite{jeh2003scaling}.  
Both of these works suggest the heuristic of propagating from the node with the largest unpropagated estimate. Our work is different because we are interested in estimating the values $\pi(u,v)$ for a single target $v$, while earlier work was concerned with the values for a single source $u$.  Because of this, our analysis is completely different, and we are able to prove running time bounds.

To the best of our knowledge, the only previous work to consider the problem of computing personalized PageRank to a target node was by Benczur et al. \cite{benczur2005spamrank}, where it is used as one phase of an algorithm to identify link-spam.  They observe that a node $v$'s global PageRank is the average over all nodes $u$ of $\pi(u,v)$.  Thus to determine how a node $v$ achieves its global PageRank score, they propose we first find the nodes $u$ with a high value of $\pi(u,v)$.  Once that set has been found, it can be analyzed to determine if it looks like an organic set of nodes or an artificial link-farm.  To compute the values of $\pi(u,v)$ for each $v$, they propose taking random walks from every source node and do not consider other methods.

\section{Preliminaries}
We are given a directed or undirected graph $G=(V,E)$.  For now we assume $G$ is unweighted, but in section \ref{sec:weighted_graph} we show how our algorithm and theorems generalize easily to weighted graphs.  We define $\out(u) = \{w : (u,w) \in E\}$ and $\inn(u) = \{w : (w,u) \in E\}$. We are given a parameter $\alpha$ which determines the expected length of a random walk, $\frac{1}{\alpha}$.  For $u,v \in V$, we define personalized PageRank $\pi(u,v)$ to be the fraction of time we spend at $v$ on the following random walk: we start at $u$ and at each step with probability $\alpha$ we halt, while with probability $1-\alpha$ we transition to a random out-neighbor of the current node.  We refer to $\alpha$ as the teleport probability because another description of the Markov chain is the following: the process never halts, and at each step with probability $\alpha$ we teleport back to $u$ and continue from there, while with probability $1-\alpha$ we transition to a random out-neighbor of the current node.  

There may be dead end nodes with no out-neighbors in the graph, so for convenience we introduce an artificial sink node with a self-loop and introduce an artificial edge to the sink from each dead end node.  Alternatively, we could have artificially added a self-edge to each dead end node, or said that the walk should halt when it reaches a dead end node.  These alternatives result in a slightly different boundary case or normalization, but the exact choice doesn't matter significantly.  

In the original PageRank paper \cite{page1999pagerank}, the authors propose that when the random walk for computing PageRank teleports, the resulting node could be chosen from an arbitrary distribution.  We focus on the case when the distribution has a single point of support, because that is the case relevant to our applications.  The PageRank function is linear in the personalization distribution, as shown in \cite{jeh2003scaling}, so computing PageRank on single-point distributions is sufficient for computing it on arbitrary personalization distributions.

In the worst case, the target node $v$ might have an edge from every other node in the graph, so we must do $\Omega(n)$ work even for a rough approximation.  To parameterize the difficulty of the problem for a given $v$, we define
 \[ D_v(x) = \sum_{u:\pi(u,v) > x} \pn{\ind{u} + \log(n)}.\]
This parameter captures the idea that for each node $u$ which has a large personalized PageRank to $v$, we must consider all of $u$'s in-neighbors to see if any of them also have a large personalized PageRank to $v$.  The $\log(n)$ term captures the cost of popping from a priority queue.

In evaluating our approximation we consider additive pointwise error (the $L^{\infty}$ norm).   Given error threshold $\epsilon$, we seek an estimate $s(u)$ for each $u$ such that
\[ \max_{u \in V} \abs{s(u) - \pi(u,v) }  < \epsilon.\]
We choose this error measure because in applications we are often only interested in the nodes $u$ with a large value of $\pi(u,v)$, and we don't care if there are a large number of nodes with very small values of $\pi(u,v)$ which have been estimated to be 0.  For efficiency we want the resulting estimate vector $s$ to be sparse unless $\epsilon$ is very small, and this norm allows for a sparse estimate vector.

\section{A Recurrence for Personalized PageRank}
Our algorithm is based on a recurrence equation that relates the value of $\pi(u,v)$ to the values of $\pi(w,v)$ for $w \in \out(u)$.  To derive this recurrence, it is convenient to think about the number of times we visit $v$ on a random walk rather than the fraction of time we spend at $v$.  The number of times we visit $v$ is proportional to the fraction of time $\pi(u,v)$ because over a large number of walks, the number of times we visit $v$ will be the fraction of time we spend at $v$ multiplied by the average length of a walk, $\frac{1}{\alpha}$.  A random walk from $u$ begins by either teleporting immediately or by transitioning to a random neighbor, so the expected number of times we reach $v$ from $u$ is the probability of not teleporting immediately times the average expected number of times we reach $v$ from an out-neighbor of $u$. Thus personalized PageRank satisfies the recurrence
  \begin{align} \label{eq_power_iteration}
 \pi(u,v) = (1-\alpha) \frac{1}{\bars{\out(u)}} \sum_{w \in \out(u)} \pi(w,v)
 + \begin{cases} \alpha, & u=v \\ 0, & u \neq v \end{cases}.     
  \end{align}
We add $\alpha$ when $u=v$ because a walk from $v$ clearly visits $v$ on its first step regardless of what happens next, and this visit corresponds to an $\alpha$-fraction of an average walk.  This equation appears as equation (10) in \cite{jeh2003scaling}, where the authors give an alternate proof using linear algebra.

\section{The Priority Queue Algorithm}
\label{sec:algorithm}
Given a  target node $v$, our algorithm is based on the idea of propagating  updates outwards from $v$.  We maintain for each node $u$ a score $s(u)$ which estimates $\pi(u,v)$ from below and improves as the algorithm progresses.  Using the recurrence of equation \eqref{eq_power_iteration}, we see that when we update our estimate $\pi(w,v)$ for some node $w$, we need to update our estimate of $\pi(u,w)$ for each $u \in \inn(w)$.  Hence the basic update step of the algorithm is to choose a node $w$ and increase the score of each in-neighbor $u$ by  $(1-\alpha) \frac{1}{\outd{u}} s(w)$.  Since we might propagate a node $w$'s score more than once, it is important that we only propagate the part of $w$'s score which changed since the last time $w$'s score was propagated.  We let $p(w)$ denote the difference between $w$'s current score $s(w)$ and $w$'s score when its score was last propagated. We use a priority queue $Q$ ordered by priority $p(w)$ so we can easily find the node with the largest value of $p(w)$.  The complete algorithm is shown in Algorithm \ref{fig:algorithm}: as long as some node has priority above a minimum threshold, we pop off the node with the greatest priority and propagate its score to its in-neighbors. 
\begin{algorithm}
   \caption{
Computing personalized PageRank to a target.}
    \label{fig:algorithm}
\begin{algorithmic}
\Require digraph $G=(V,E)$,teleport probability $\alpha$,target vertex $v$, error tolerance $\epsilon$
\Ensure Approximation $s:V \to [0,1]$ to personalized PageRank such that for all $u$, $\abs{\pi(u, v) - s[u]} < \epsilon$ \\
\State  $s[v] = p[v] = \alpha$
\State  $q$ = Max Priority Queue on $V$ ordered by key $p$
\While {$\text{$q$.maxPriority()} > \alpha \cdot \epsilon$}
\State    $w$ = $q$.popMaxElement()    
\For {$u$ in $w$.inNeighbors()}
\State      $\Delta s =  (1-\alpha) \frac{p[w]}{ u.\text{outDegree}}$
\If {$u$ not in $s$}
\State $s[u] = p[u] = 0$
\EndIf
\State      $s[u] = s[u] + \Delta s$
\State      $q$.increasePriority($u$, $p[u] + \Delta s$)
\EndFor
\State      $p[w] = 0$
\EndWhile
\end{algorithmic}
\end{algorithm}    

An example run of 6 iterations of the algorithm is shown in Figure \ref{example_iterations}.
\begin{figure*}
  \centering
  \includegraphics[width=0.4\textwidth]{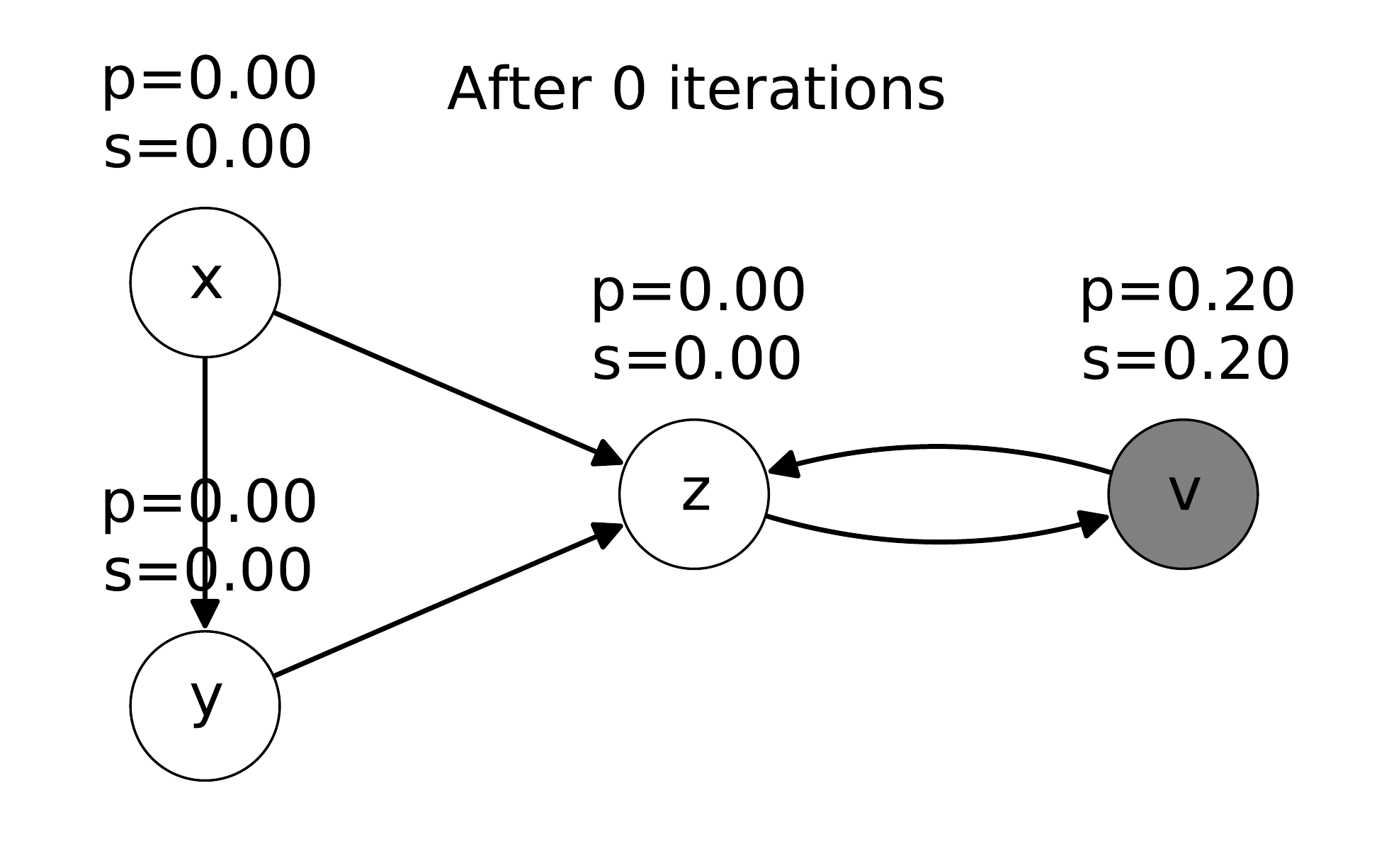}
  \includegraphics[width=0.4\textwidth]{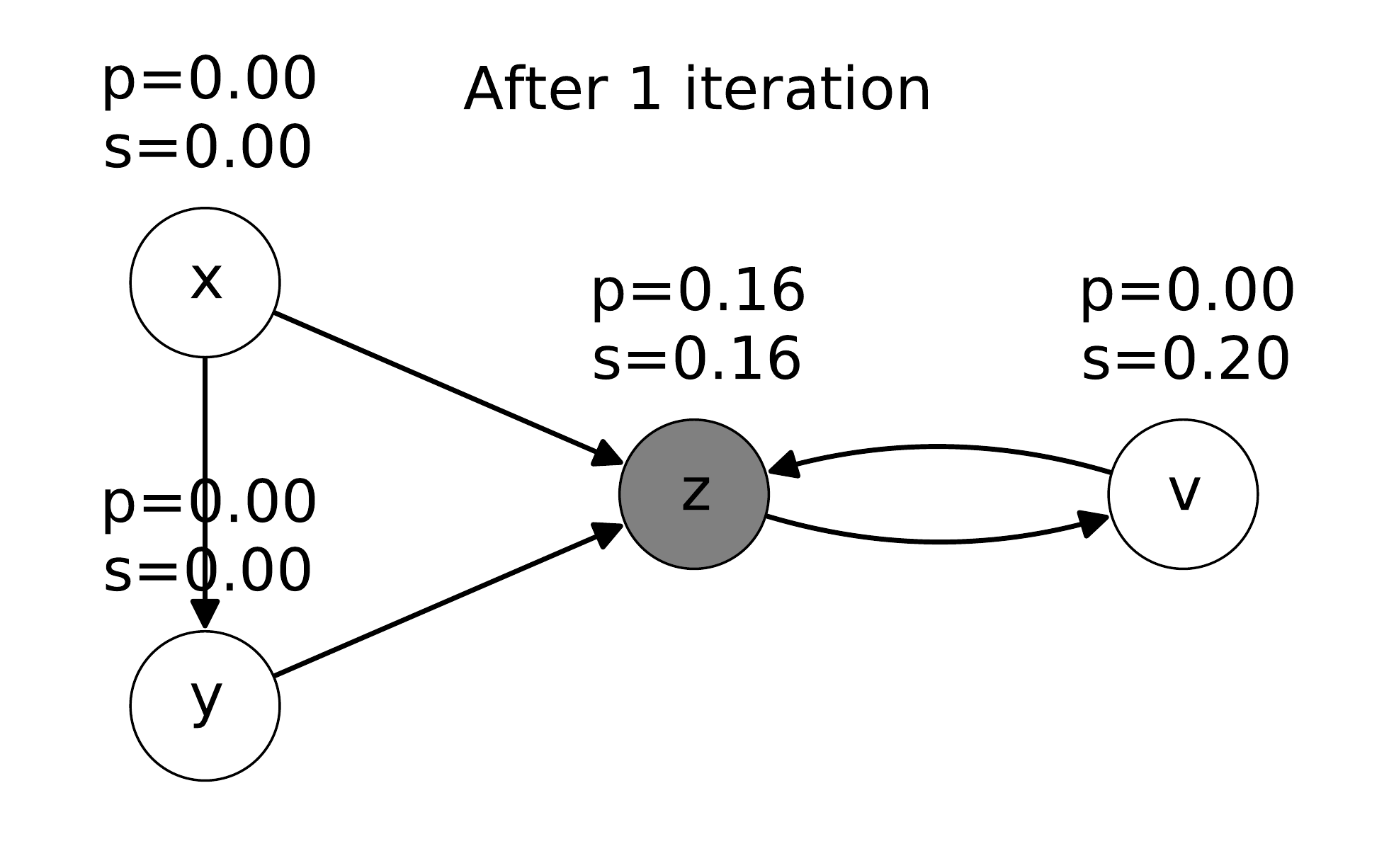}
  \includegraphics[width=0.4\textwidth]{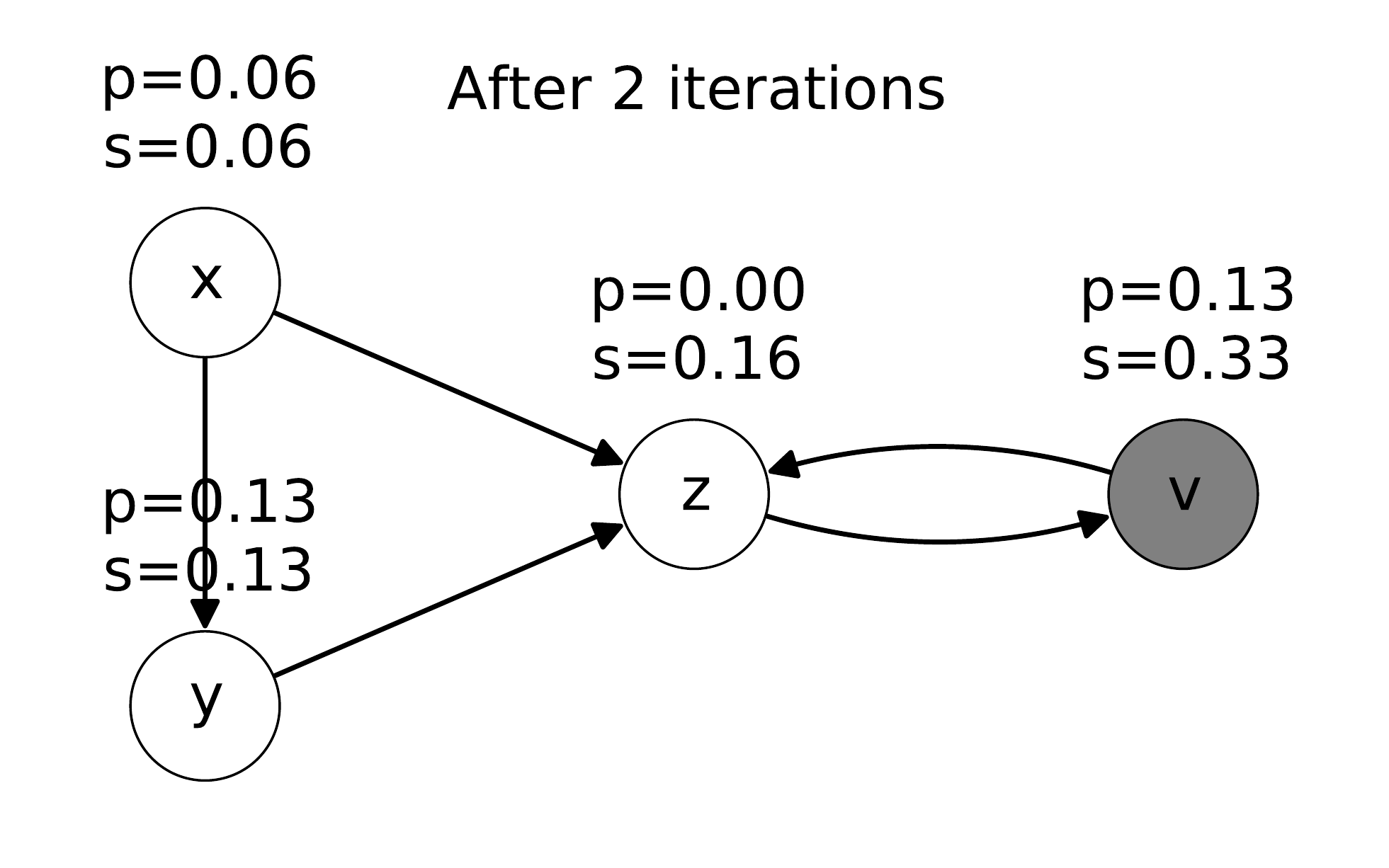}
  \includegraphics[width=0.4\textwidth]{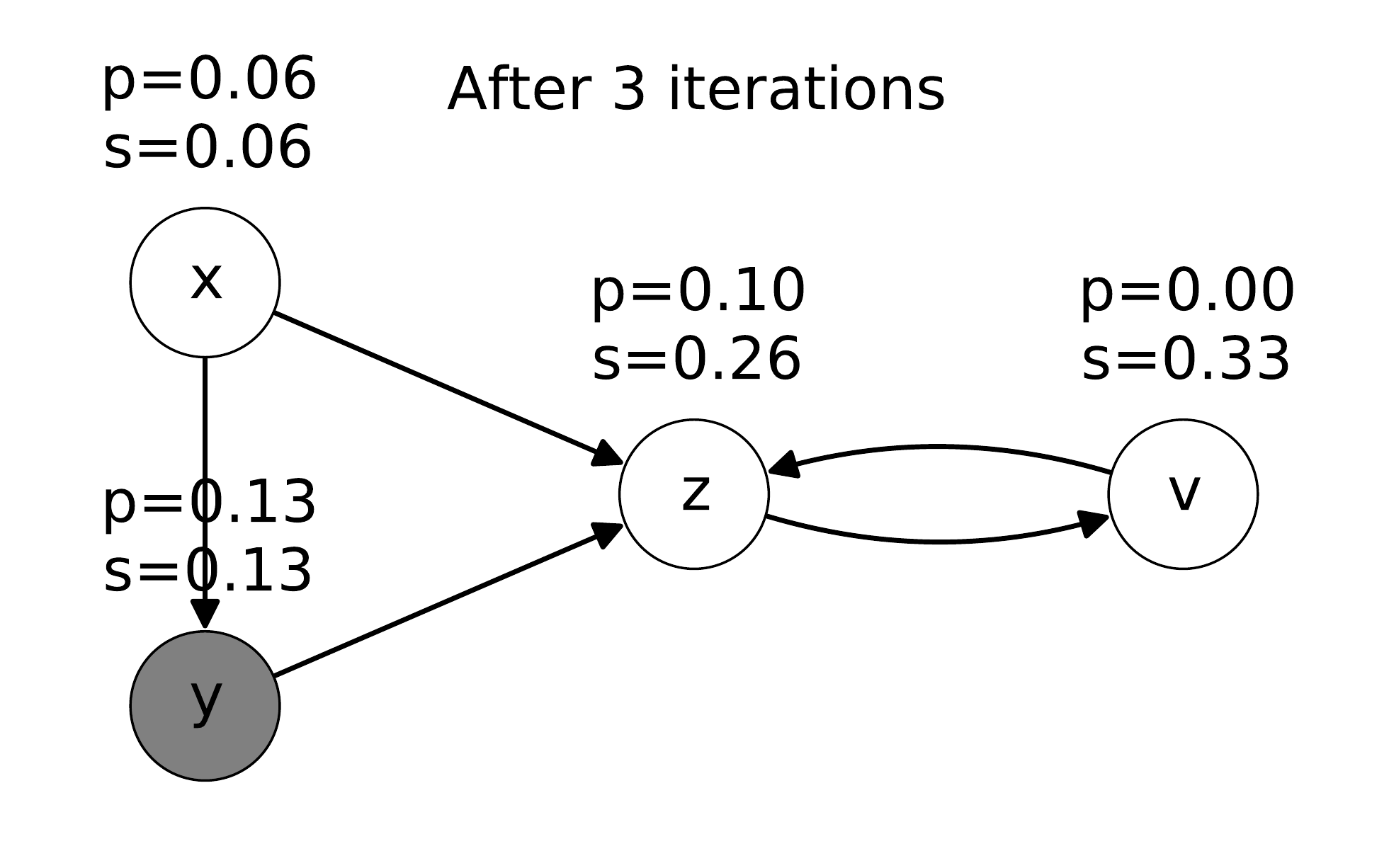}
  \includegraphics[width=0.4\textwidth]{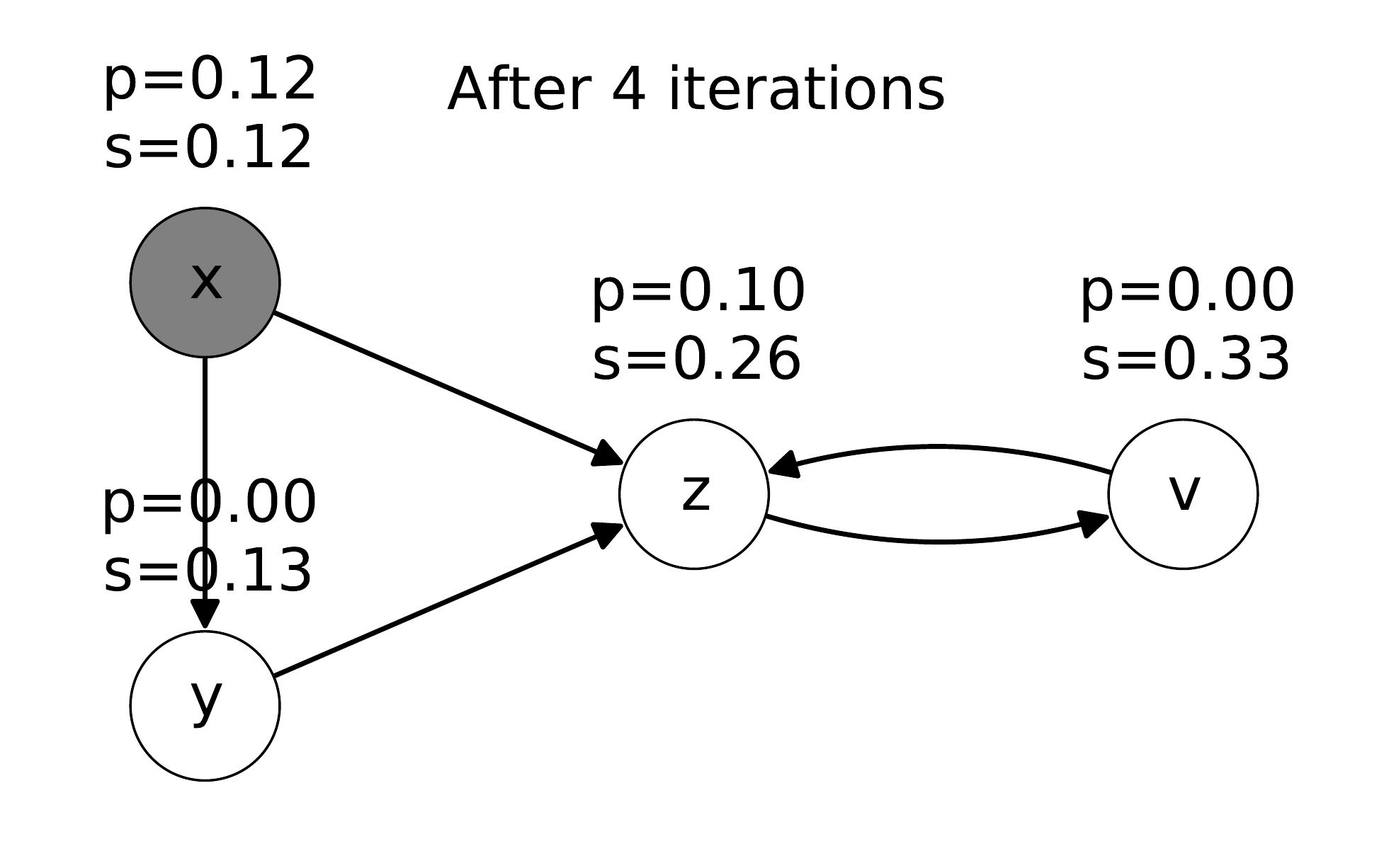}
  \includegraphics[width=0.4\textwidth]{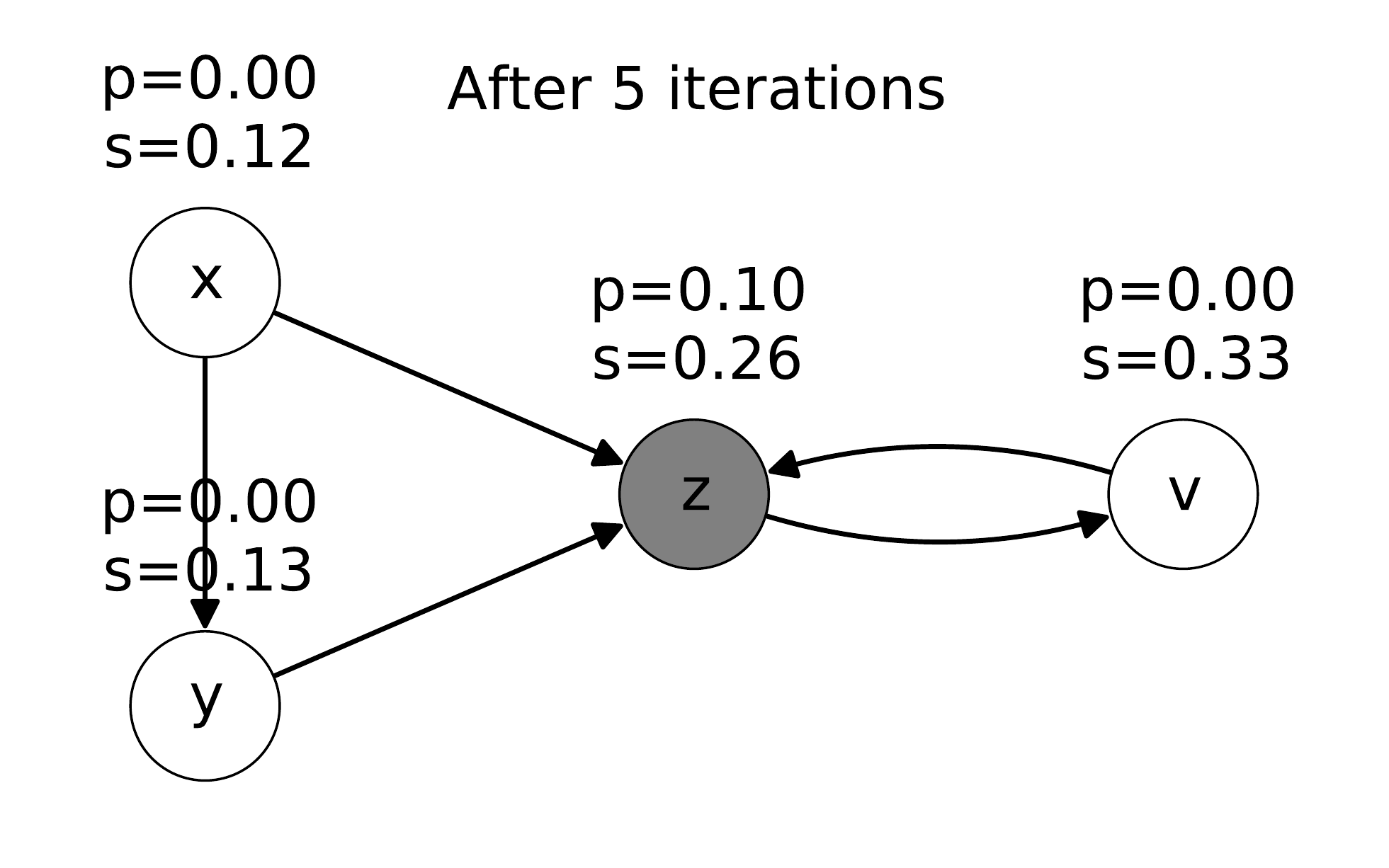}
  \caption{The first six iterations of the priority queue algorithm run on a simple graph with four nodes and target node v.  The node which will propagate its priority next is shown with a dark background.  The score $s$ for a node $u$ is our current estimate of $\pi(u,v)$, and the priority $p$ is the amount of unpropagated score.}
  \label{example_iterations}
\end{figure*}

\subsection{Error Analysis}
\label{sec:error_analysis}
One key question was the priority threshold at which we should stop popping nodes.  Initially we considered threshold $\epsilon$, but this threshold is not small enough to ensure that all errors are less than $\epsilon$.  We now show that $ \alpha \epsilon $ is a sufficient threshold, and our experiments show that it is tight.
\begin{thm}[Correctness]\label{thm-correctness} 
  When the priority queue algorithm is run until all priorities are less than $\alpha \epsilon$, the resulting score vector $s$ 
satisfies $\abs{s(u) - \pi(u,v)} < \epsilon$ for all $u \in V$.
\end{thm}
\begin{proof}
  After the algorithm has run to completion, let $u$ be the node with the greatest additive error and let $E=\abs{\pi(u,v)-s(u)}$ be its error, so for all nodes $w$, $\abs{\pi(w,v) - s(w)} \leq E$.   Recall equation \eqref{eq_power_iteration}:
  \begin{align*}
    \pi(u,v) = (1-\alpha) \frac{1}{\bars{\out(u)}} \sum_{w \in \out(u)} \pi(w,v)
    + \begin{cases} \alpha, & u=v \\ 0, & u \neq v \end{cases}.     
  \end{align*}
When the algorithm has completed, a node $u$'s score is equal to the sum of the amount of score it has received from each out-neighbor.  An out-neighbor $w$ has final score $s(w)$ and final un-propagated score $p(w)$, so the amount propagated is $s(w)- p(w)$.  This gives us the following:
\[ s(u) = (1-\alpha) \frac{1}{\bars{\out(u)}} \sum_{w \in \out(u)} \pn{s(w) - p(w)}
+ \begin{cases} \alpha, & u=v \\ 0, & u \neq v \end{cases} 
\]
where $p(w) < \alpha \epsilon$ is the part of $w$'s score which has not been propagated back to $u$.  Subtracting these two equations, we see that
\begin{align*}
 E &= \abs{\pi(u,v) - s(u)} \\
   &=  (1-\alpha) \frac{1}{\bars{\out(u)}} \sum_{w \in \out(u)} \abs{\pi(w,v) - s(w)} \\
 & \quad + (1-\alpha) \frac{1}{\bars{\out(u)}} \sum_{w \in \out(u)} p(w) \\
  & \leq (1-\alpha) E + (1-\alpha)  \alpha \epsilon
\end{align*}
where we've used the fact that for all nodes $w$, $\abs{\pi(w,v) - s(w)} \leq E$ and $p(w) < \alpha \epsilon$.
Isolating the error $E$ we conclude that
\[ E \leq (1-\alpha) \epsilon < \epsilon. \]
\end{proof}

\subsection{Average Running Time}
\label{sec:average_runtime}
Next we analyze the running time of this algorithm.  In the worst case, our target node $v$ could have a high personalized PageRank from every other node, forcing us to consider the entire graph and do $\Omega(m)$ work.  Thus to give a useful bound on the running time, we give both an average case analysis and a worst-case parameterized analysis.  First we analyze the priority queue algorithm in the average case where the target node $v$ is chosen uniformly at random.
\begin{thm} \label{thm:average_case}
  Let an arbitrary graph $G$, additive error tolerance $\epsilon$, and teleport probability $\alpha$ be given. Let $n$ be the number of nodes in $G$ and $m$ be the number of edges.  If $v$ is chosen uniformly at random from $V$, then the priority queue algorithm runs in expected time $O \pn{\frac{1}{\alpha \epsilon} \pn{\frac{m}{n} + \log(n)}}$ steps. 
\end{thm}
\begin{proof}
  Suppose we ran the algorithm once for every $v \in V$.  When the target node is $v$, the number of times a node $u$ can be popped from the queue is at most $\floor{\frac{\pi(u,v)}{\alpha \epsilon}}$, since its priority decreases by at least $\alpha \epsilon$ each time it is popped and the total accumulated priority is at most $\pi(u,v)$.  The time to propagate the score from a node $u$ is $O(\ind{u})$ steps since each of its in-neighbors must receive some of its score.  We also must do $O(\log(n))$ work to pop the maximum node from the priority queue.  Thus the running time for all $n$ nodes is at most
  \begin{align*}
    &\sum_{v \in V} \sum_{u \in V} \frac{\pi(u,v)}{\alpha \epsilon} O\pn{\bars{\inn{u}} + \log(n)} \\
      &\quad =\sum_u \sum_v \frac{\pi(u,v)}{\alpha \epsilon} O\pn{\bars{\inn{u}} + \log(n)} \\
      &\quad = \sum_u \frac{1}{\alpha \epsilon} O\pn{\bars{\inn{u}} + \log(n)}\\
      &\quad = O\pn{\frac{m + n \log(n)}{\alpha \epsilon}},
  \end{align*}
and the average running time per node is as claimed.
\end{proof}

In the appendix we prove a bound which is tighter for $\epsilon > \frac{1}{n}$ in the case when the personalized PageRanks from each source follow a power law.

Because we perform a large number of increase-priority operations on the priority queue, the best asymptotic time is achieved by using a Fibonacci heap for the priority queue.  With a Fibonacci heap \cite{fredman1987fibonacci}, we can increase a node's priority in constant amortized time, so in the above analysis the cost of $O(\ind{u})$ increase-priority operations is $O(\ind{u})$.  In our experiments we use a standard binary-heap priority queue for simplicity.  

Also note that our average running time analysis did not use the fact that we are using a priority queue.  The same time bound would hold if we simply maintained the set of nodes $u$ with $\pi(u) > \alpha \epsilon$ and repeatedly popped an arbitrary element of this set.  This is an alternative implementation of the algorithm which avoids the cost of the queue.  By removing the cost of the priority queue from the above analysis we see that this alternative runs in time
\[O \pn{\frac{1}{\alpha \epsilon} \frac{m}{n}}. \]
Our parameterized running time analysis does use the priority queue property to improve the dependence on $\epsilon$ from $\frac{1}{\epsilon}$ to $\log \pn{\frac{1}{\epsilon}}$.

{\bf Comparison with Monte Carlo} In \cite{benczur2005spamrank}, the authors suggest computing values of $\pi(u,v)$ to a given target $v$ by taking Monte Carlo walks from every other node.  As shown in \cite{fogaras2005towards} using the Chernoff bounds, computing an $\epsilon$ approximation of $\pi(u,v)$ for a single source $u$ and with failure probability $\delta$ takes
\[\Theta \pn{ \frac{1}{\epsilon^2} \log\pn{\frac{1}{\delta}}} \]
steps.  
If the graph is sparse enough or $\epsilon$ is small enough, our average time bound of 
\[  O\pn{\frac{1}{\alpha \epsilon} \frac{m}{n}} \] 
is better than this.  Thus to compute an $\epsilon$ approximation of personalized PageRank for all pairs of nodes, running the priority queue algorithm to every node is a viable alternative to running Monte Carlo from every node.  

\subsection{Parameterized Running Time}
\label{sec:parameterized_runtime}
Next we give a paramterized bound that applies to arbitrary graphs and arbitrary target node $v$.  As in the preliminaries section, we define
\[D_v(x) = \sum_{u : \pi(u,v) > x} \pn{\ind{u} +  \log(n)}\]
to capture the difficulty of computing personalized PageRank to the target $v$.
\begin{thm} \label{parameterized_case}
If the priority queue algorithm is run with teleportation probability $\alpha$, target node $v$, and additive error $\epsilon$, it takes time
\[ O \pn{\frac{D_v(\alpha \epsilon)}{\alpha} \log \pn{\frac{1}{\epsilon \alpha}}} \]
\end{thm}
\begin{proof}
  We divide the execution of the algorithm into $\log_2 \pn{\frac{1}{\alpha \epsilon}}$ stages, where in stage $i$ it pops nodes with priority greater $\frac{1}{2^i}$ until all nodes have priority less than $\frac{1}{2^i}$.  After stage $i$ has completed, by Theorem \ref{thm-correctness}, the difference between the score of a node $u$ and the true value of $\pi(u,v)$ is at most $\frac{1}{2^i \alpha}$.  This implies that each node $u$ can be popped at most $\frac{2}{\alpha}$ times in each stage, since each pop in stage $i+1$ decreases the difference between $s(u)$ and $\pi(u,v)$ by at least $\frac{1}{2^{i+1}}$.  Each time a node $u$ is popped we do $O(\ind{u})$ work increasing priorities and $O(\log(n))$ work popping the node from the priority queue.  This gives us a running time of
\begin{align*}
  &O \pn{\frac{2}{\alpha} D_v \pn{\frac{1}{2}} + \frac{2}{\alpha} D_v \pn{\frac{1}{2^2}}  + \cdots + \frac{2}{\alpha} D_v \pn{\frac{1}{2^{\log_2 \pn{\frac{1}{\alpha \epsilon}}}}}} \\
&=  O \pn{\frac{1}{\alpha} D_v(\alpha \epsilon) \log \pn{\frac{1}{\epsilon \alpha}}}.
\end{align*}
\end{proof}


\subsection{Extension to Weighted Graphs} \label{sec:weighted_graph}
We assume that the graph is unweighted for simplicity, but our algorithm extends immediately to the case of a weighted graph, in which $\text{weight}[u][w]$ is proportional to the probability of transitioning from node $i$ to node $j$ on a random walk.  In this case, the change in score in Algorithm \ref{fig:algorithm} should become
\[\Delta s =  (1-\alpha) \frac{p[w] \text{weight}[u][w]}{ u.\text{weightedOutDegree}}\]
where $u.\text{weightedOutDegree}$ is  defined as $\sum_w \text{weight}[u][w]$.  Similarly, the power iteration equation \eqref{eq_power_iteration} should become
\begin{align*}
 \pi(u,v) &= (1-\alpha) \frac{1}{d_{\text{OUT}}(u)} \sum_{w \in \out(u)} \text{weight}[u][v] \cdot \pi(w,v)
\\& \quad + \begin{cases} \alpha, & u=v \\ 0, & u \neq v \end{cases}.     
  \end{align*}
where $d_{\text{OUT}} = \sum_{w \in \out(u)} \text{weight}[u][v]$ is the weighted out-degree of $u$. 
All the proofs can be modified similarly.  The theorem statements remain the same.

\section{Experiments}
\label{sec:experiments}
For our experiments, we used a part of the Twitter follower graph with 
5.3 million nodes and 389 million edges.  We ran an experiment for each setting of parameters in the Cartesian product of teleport probability $\alpha \in \{0.1, 0.2\}$ and additive error $\epsilon \in \{10^{-4}, 10^{-5}, 10^{-6}\}$.  For each experiment we chose 100  target nodes $v$ uniformly at random and ran the priority queue algorithm.  Since nodes with high global PageRank might be targets more often than other nodes, we repeated the above setup sampling 100 target nodes with probability equal to their global PageRank. We measured the number of steps (defined as the number of times we updated some node's priority in the inner loop), the change in wall-clock time, and the maximum error.  We measured the maximum error by running power iteration, equation \eqref{eq_power_iteration}, until convergence and comparing the result pointwise with the result of the priority queue algorithm.  

We first note that our error analysis is tight.  For efficiency, we want to do as few operations as possible to reach our desired error tolerance $\epsilon$.  If our empirical error was much lower than our target error $\epsilon$, it would indicate that we were wasting effort achieving an accuracy which is finer than required.

However, on the Twitter graph there are nodes with empirical error $0.85\epsilon$, which is quite close to our proven bound of $\epsilon$.  A histogram of the empirical errors for $\epsilon=10^{-6}$ and targets sampled from the global PageRank distribution is shown in Figure \ref{fig:errors}.  Notice that for these parameters the empirical error is often more than 50\% of the proven bound $\epsilon$.
\begin{figure}
  \centering
  \includegraphics[width=0.5\textwidth]{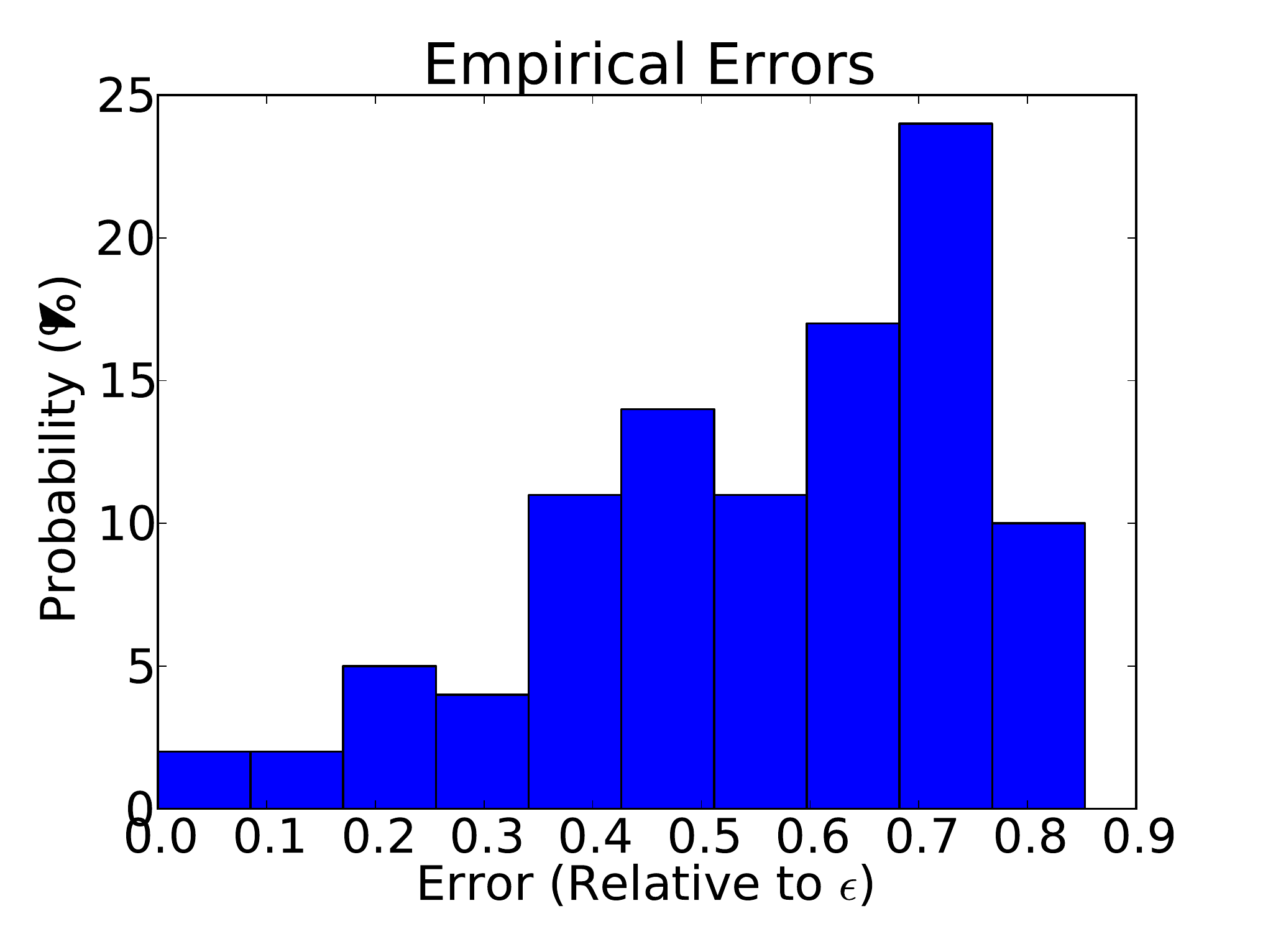}
  \caption{The empirical error of our algorithm after convergence.  To obtain this data, we set teleport probability $\alpha=.1$ and error threshold $\epsilon=10^{-6}$, choose 100 target nodes $v$ with probability equal to their global PageRank, ran the priority queue algorithm to obtain scores $s$, and then computed the empirical error $\max_{u \in V} \abs{ \pi(u,v) - s(u)}$. The x-axis is empirical error divided by $\epsilon$.  Notice that most nodes have an error which is a large fraction of $\epsilon$, showing that our error bound is tight in practice.}

  \label{fig:errors}
\end{figure}

We compare our actual average running time to the bound from Theorem \ref{thm:average_case}, $\frac{1}{\alpha \epsilon} \frac{m}{n}$ steps, and find that the algorithm actually runs faster than the bound requires.  For $\alpha=0.2$ and all three values of $\epsilon$, the algorithm uses less than 3\% 
of the number of steps the bound allows.  For $\alpha=0.1$, and all values of $\epsilon$ the algorithm uses less than 20\% 
of the number of steps the bound allows.  A histogram of the running times is shown in figure \ref{fig:average_steps}.  
Notice that the step-axis is log-scale, and most nodes use far fewer steps than the bound represented by the vertical line allows.
\begin{figure}
  \centering
  \includegraphics[width=0.5\textwidth]{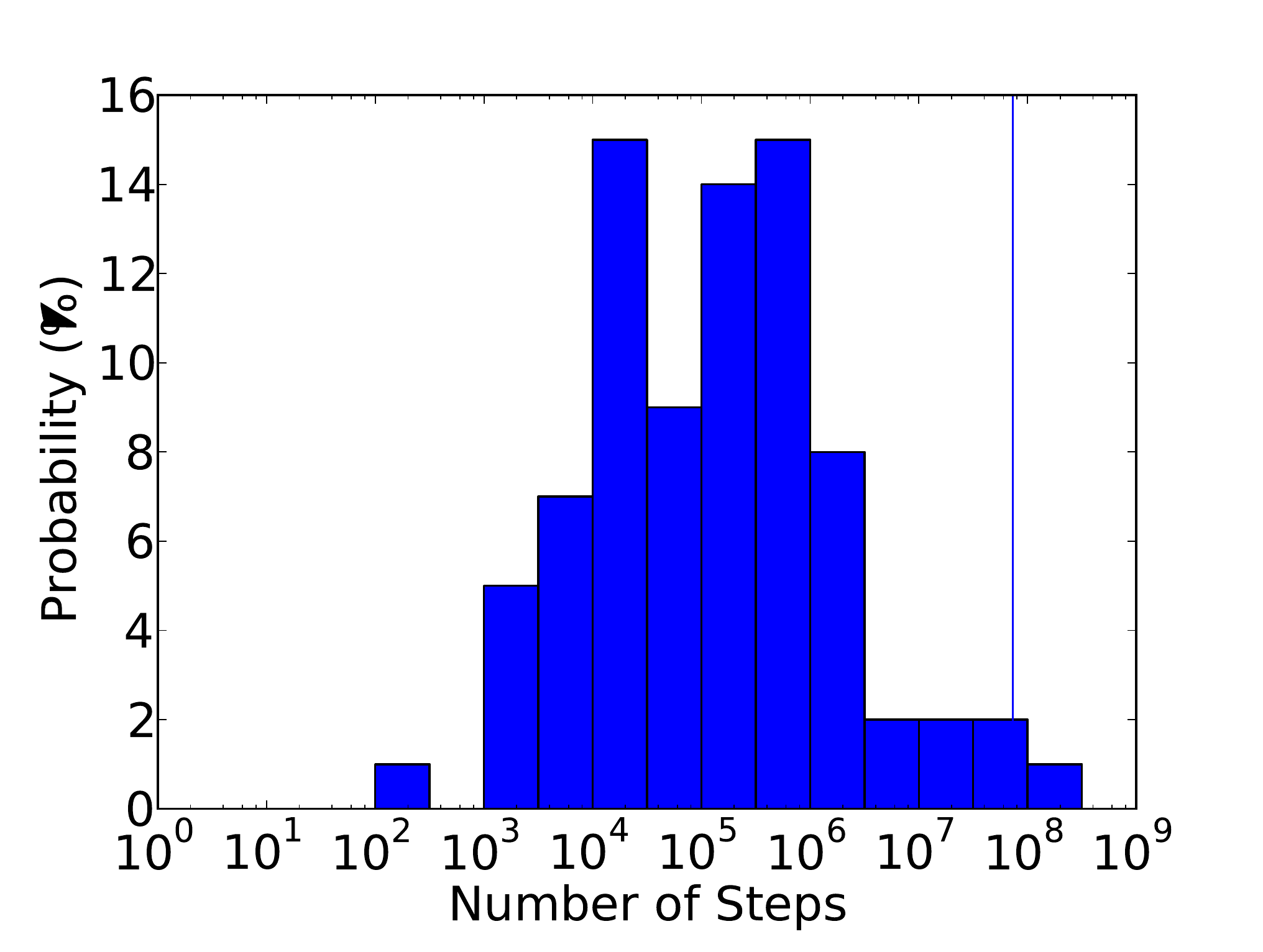}
  \caption{  The number of steps required to reach convergence, defined as the number of times we updated some node's priority in the inner loop.  To obtain this data, we set teleport probability $\alpha=.1$ and error threshold $\epsilon=10^{-5}$, choose 100 target nodes $v$ uniformly at random, and ran the priority queue algorithm. The vertical line indicates the average running time bound of Theorem \ref{thm:average_case}. Notice that the x-axis is log-scale, so most nodes require many fewer steps than the bound allows.
    } 
  \label{fig:average_steps}
\end{figure}

Our parameterized analysis shows that the number of steps needed is at most $O \pn{\frac{D_v(\alpha \epsilon)}{\alpha} \log \pn{\frac{1}{\epsilon \alpha}}}$.  To measure how tight this is, we compared the number of steps taken to $D_v$.  To use more adversarial $v$, we sampled $v$ from the global PageRank distribution instead of uniformly at random, so $v$ with high global PageRank will be chosen more often.  We found that in practice $D_v(\alpha \epsilon)$ is an excellent predictor for the number of steps, and that the constant of proportionality in practice is much less than $\frac{1}{\alpha} \log \pn{\frac{1}{\epsilon \alpha}}$.  For example, with $\alpha=0.1$ and $\epsilon=10^{-5}$, the proven ratio between step count and $D_v(\alpha \epsilon)$ is $200$, but in our experiment the average ratio is less than 4.  The distribution of ratios is shown in Figure \ref{fig:D_v}.    Note that for most nodes, the number of steps taken is within a factor of 2 of $D_v(\epsilon \alpha)$ even though the absolute number of steps varies on an exponential scale, as shown in Figure \ref{fig:average_steps}.
\begin{figure}
  \centering
  \includegraphics[width=0.5\textwidth]{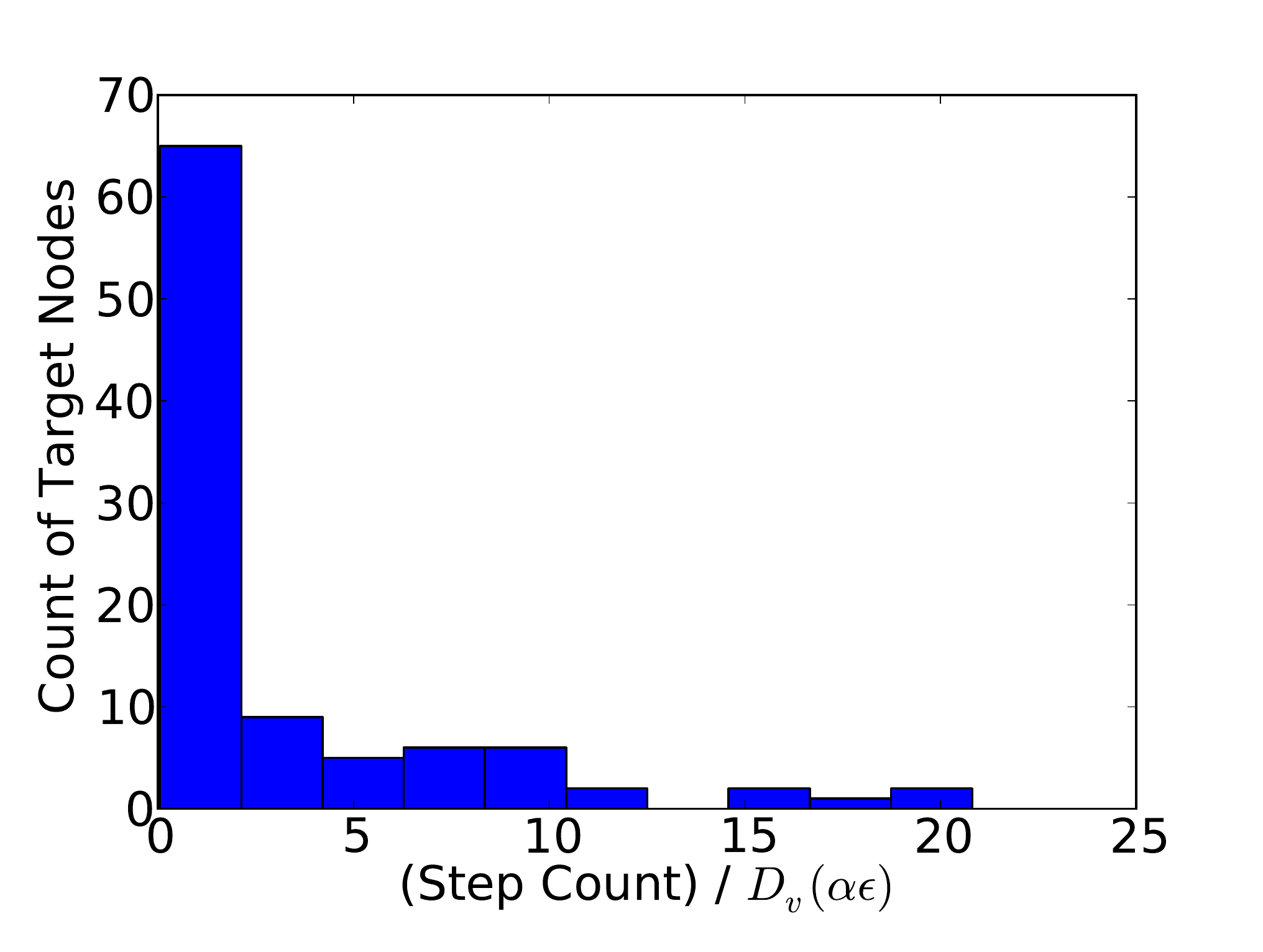}
  \caption{The number of steps required to reach convergence compared to the parameter $D_v(\alpha \epsilon)$.  To obtain this data, we set teleport probability $\alpha=.1$ and error threshold $\epsilon=10^{-5}$, choose 100 target nodes $v$ sampled from the PageRank distribution, and ran the priority queue algorithm.  Notice that for most nodes, the number of steps taken is within a fraction of 2 of $D_v(\epsilon \alpha)$ even though the absolute number of steps varies on an exponential scale, as shown in Figure \ref{fig:average_steps}. } 
  \label{fig:D_v}
\end{figure}

Now we compare our algorithm's performance to power iteration and observe the benefit of only visiting a set of nodes around $v$.  It can be shown that equation \eqref{eq_power_iteration} is a contraction mapping with contraction ratio $1- \alpha$.  Thus one alternative to our priority queue algorithm is to apply equation \eqref{eq_power_iteration} repeatedly. Using the contraction map property, to guarantee additive error $\epsilon$ we must do $\log_{1-\alpha}{\frac{1}{\epsilon}}$ iterations.  In our experiment, we also computed personalized PageRank to each target using power iteration in order to measure the empirical error.  We found that each iteration of applying equation \eqref{eq_power_iteration} took 3.9 seconds, and this was stable over a large number of iterations (the graph is too large for the processor cache to help).  Because our algorithm explores only a neighborhood around $v$, it is often much more efficient than power iteration. For example, when $\alpha=0.1$ and $\epsilon=10^{-4}$, our algorithm took 0.2 seconds on average, which is 1700 times faster than the 87 iterations needed to guarantee at most $\epsilon$ error. For smaller $\epsilon$, our algorithm is forced to consider more of the graph, so its relative advantage diminishes.  For the smallest value of $\epsilon$ we tried, $\epsilon=10^{-6}$, our algorithm took 30 seconds on average, while power iteration takes $3.8 \log_{1-\alpha}{\frac{1}{\epsilon}} \approx 500$ seconds.  A table of running times for the two algorithms is shown in Figure \ref{fig:runtime_table}.
  \begin{figure}[tb]
    \centering
    \begin{tabular}{c|c|c}
      $\epsilon$ & Priority Queue Algorithm (s) & Power Iteration (s)\\
     \hline
     $10^{-4}$ & 0.20 & 330\\
     \hline
     $10^{-5}$ &  1.2 & 410\\
     \hline
     $10^{-6}$ & 29 & 500
    \end{tabular}
    \caption{The average wall-clock running time of our algorithm compared to power iteration.  To obtain the first column, we set teleport probability $\alpha=.1$, choose 100 target nodes $v$ uniformly at random, and ran the priority queue algorithm until completion.  To obtain the second column, we measured the time for an average iteration of applying equation \eqref{eq_power_iteration} and multiplied by $\log_{1-\alpha} \pn{\frac{1}{\epsilon}}$, the number of iterations needed for accuracy $\epsilon$.  Notice that by propagating the largest changes first, our algorithm is much faster than power iteration.}
    \label{fig:runtime_table}
  \end{figure}
As $\epsilon$ tends to zero, our algorithm's performance degrades gracefully to within a constant factor of the performance of power iteration, as proven in Theorem \ref{parameterized_case}.

\section{Acknowledgements}
This work was supported in part by the DARPA xdata program, by grant \#FA9550-
12-1-0411 from the U.S. Air Force Office of Scientific Research
(AFOSR) and the Defense Advanced Research Projects Agency (DARPA), and
by NSF Award 0915040.  One of the authors was supported by the National Defense Science \& Engineering Graduate Fellowship (NDSEG) Program.  We would like to thank Rishi Gupta for helpful conversations.

\balance

\bibliographystyle{abbrv}
\bibliography{susceptibility}

\begin{thebibliography}{10}

\bibitem{andersen2008trust}
R.~Andersen, C.~Borgs, J.~Chayes, U.~Feige, A.~Flaxman, A.~Kalai, V.~Mirrokni,
  and M.~Tennenholtz.
\newblock Trust-based recommendation systems: an axiomatic approach.
\newblock In {\em Proceeding of the 17th international conference on World Wide
  Web}, pages 199--208. ACM, 2008.

\bibitem{andersen2007local}
R.~Andersen, C.~Borgs, J.~Chayes, J.~Hopcraft, V.~S. Mirrokni, and S.-H. Teng.
\newblock Local computation of pagerank contributions.
\newblock In {\em Algorithms and Models for the Web-Graph}, pages 150--165.
  Springer, 2007.

\bibitem{Andersen:2008:RPL:1451983.1452000}
R.~Andersen, C.~Borgs, J.~Chayes, J.~Hopcroft, K.~Jain, V.~Mirrokni, and
  S.~Teng.
\newblock Robust pagerank and locally computable spam detection features.
\newblock In {\em Proceedings of the 4th international workshop on Adversarial
  information retrieval on the web}, AIRWeb '08, pages 69--76, New York, NY,
  USA, 2008. ACM.

\bibitem{andersen2006local}
R.~Andersen, F.~Chung, and K.~Lang.
\newblock Local graph partitioning using pagerank vectors.
\newblock In {\em Foundations of Computer Science, 2006. FOCS'06. 47th Annual
  IEEE Symposium on}, pages 475--486. IEEE, 2006.

\bibitem{bahmani2010fast}
B.~Bahmani, A.~Chowdhury, and A.~Goel.
\newblock Fast incremental and personalized pagerank.
\newblock {\em Proceedings of the VLDB Endowment}, 4(3):173--184, 2010.

\bibitem{benczur2005spamrank}
A.~Benczur, K.~Csalogany, T.~Sarlos, and M.~Uher.
\newblock Spamrank--fully automatic link spam detection work in progress.
\newblock In {\em Proceedings of the First International Workshop on
  Adversarial Information Retrieval on the Web}, 2005.

\bibitem{berkhin2006bookmark}
P.~Berkhin.
\newblock Bookmark-coloring algorithm for personalized pagerank computing.
\newblock {\em Internet Mathematics}, 3(1):41--62, 2006.

\bibitem{fogaras2005towards}
D.~Fogaras, B.~R{\'a}cz, K.~Csalog{\'a}ny, and T.~Sarl{\'o}s.
\newblock Towards scaling fully personalized pagerank: Algorithms, lower
  bounds, and experiments.
\newblock {\em Internet Mathematics}, 2(3):333--358, 2005.

\bibitem{fredman1987fibonacci}
M.~L. Fredman and R.~E. Tarjan.
\newblock Fibonacci heaps and their uses in improved network optimization
  algorithms.
\newblock {\em Journal of the ACM (JACM)}, 34(3):596--615, 1987.

\bibitem{jeh2003scaling}
G.~Jeh and J.~Widom.
\newblock Scaling personalized web search.
\newblock In {\em Proceedings of the 12th international conference on World
  Wide Web}, pages 271--279. ACM, 2003.

\bibitem{liben2007link}
D.~Liben-Nowell and J.~Kleinberg.
\newblock The link-prediction problem for social networks.
\newblock {\em Journal of the American society for information science and
  technology}, 58(7):1019--1031, 2007.

\bibitem{page1999pagerank}
L.~Page, S.~Brin, R.~Motwani, and T.~Winograd.
\newblock The pagerank citation ranking: bringing order to the web.
\newblock Technical report, Stanford University Database Group, 1999.

\end{thebibliography}
\appendix

\section{Average Running Time for Power Law Graphs} 
We can get a better bound on the expected running time if we assume a power law on the personalized PageRank values: suppose that for each $u$ there is some $\beta \in (0,1)$ such that if we order the nodes $v_1, \ldots, v_n$ in decreasing order of $\pi(u,v_i)$ then
\[ \pi(u,v_i) = \eta i^{-\beta} \]
for some constant $\eta$.  Since $\sum_v \pi(u,v)=1$, the value for $\eta$ is determined by $\beta$:
 \[\eta \approx \frac{1-\beta}{n^{1-\beta}}.\]  Such a power law was observed empirically on the twitter graph in \cite{bahmani2010fast} with $\beta \approx 0.75$.  For simplicity we assume that all nodes have the same exponent $\beta$.
\begin{thm}
\label{powerlaw_suscept_alg}
  For a graph $G$ in which the personalized PageRanks from each node follow a power law with exponent $\beta$, if $v$ is chosen uniformly at random from $V$, then the priority queue algorithm runs in time 
\[ O\pn{  \frac{m}{n^{\frac{1}{\beta}}} \pn{\frac{1}{\alpha \epsilon}}^{\frac{1}{\beta}}}\]
 where $m$ is the number of edges in the graph.  
\end{thm}
\begin{proof}
  Suppose we ran the algorithm once for every $v \in V$.  As in the average-case analysis proof, the running time is at most $\sum_{u,v} \floor{\frac{\pi(u,v)}{\alpha \epsilon} } \bars{\inn{u}} $.  With the power law assumption, the majority of nodes will not be popped even once because $\pi(u,v) < \alpha \epsilon$.  The largest $i$ such that $\pi(u,v_i) = \eta i^{-\beta} \geq \alpha \epsilon$ is 
\[ i_* = \pn{\frac{\eta}{ \alpha \epsilon}} ^{\frac{1}{\beta}}\]
so we only need to visit this many distinct nodes $v$ for each $u$.
The running time for all $n$ nodes is thus at most
  \begin{align*}
    \sum_v \sum_u \floor{\frac{\pi(u,v)}{\alpha \epsilon} } \bars{\inn{u}}
      & \leq \sum_u \sum_{i=1}^{i_*} \frac{\eta i^{-\beta}}{\alpha \epsilon} \bars{\inn{u}} \\
      & \approx \frac{\eta}{\alpha \epsilon} \sum_u \bars{\inn{u}} \int_1^{i_*}x^{-\beta} dx\\
      &\leq m \frac{\eta}{\alpha \epsilon} \frac{i_*^{-\beta+1}}{-\beta + 1}.
  \end{align*}
Substituting in the value of $i_*$ and $\eta$ we see that the total runtime for all $n$ target nodes is
\[ c n   \frac{m}{n^{\frac{1}{\beta}}} \pn{\frac{1}{\alpha \epsilon}}^{\frac{1}{\beta}}\]
where $c = (1-\beta)^{\pn{\frac{1}{\beta} - 1}}$ so the average running time per node is as claimed.
\end{proof}

\end{document}